\documentclass[a4paper,10pt,twoside]{article}
\pdfoutput=1

\makeatletter

\usepackage[T1]{fontenc}
\usepackage[utf8]{inputenc}
\usepackage{textcomp}

\usepackage[english]{babel}

\usepackage{enumitem}

\usepackage{color}

\usepackage{amsmath}
\usepackage{amsthm}
\usepackage{mathtools}
\usepackage{tensor}

\usepackage{url}
\usepackage[unicode,colorlinks,pdfusetitle,final]{hyperref}

\usepackage[numbers,sort]{natbib}

\IfFileExists{MinionPro.sty}{
  \usepackage[lf,medfamily,italicgreek,openg]{MinionPro}
}{
  \usepackage{lmodern}
  \usepackage{upgreek}
  \usepackage{amssymb}
  \usepackage{amsfonts}

  \let\upDelta\Updelta
  \let\upLambda\Uplambda
}

\IfFileExists{MyriadPro.sty}{
  \usepackage[lf,medfamily,onlytext]{MyriadPro}

  \usepackage[sf,bf,small]{titlesec}
  \usepackage[labelfont={sf,bf},labelsep=period]{caption}

  \let\maybesf\sffamily
}{
  \usepackage{helvet}

  \usepackage[bf,small]{titlesec}
  \usepackage[labelfont={bf},labelsep=period]{caption}

  \let\maybesf\rmfamily
}

\usepackage[final]{microtype}

\usepackage{ifdraft}

\ifdraft{
  \usepackage[notref,notcite]{showkeys}

  \usepackage[pagewise]{lineno}
  \linenumbers

  \newcommand*\patchAmsMathEnvironmentForLineno[1]{%
    \expandafter\let\csname old#1\expandafter\endcsname\csname
  #1\endcsname
    \expandafter\let\csname oldend#1\expandafter\endcsname\csname
  end#1\endcsname
    \renewenvironment{#1}%
       {\linenomath\csname old#1\endcsname}%
       {\csname oldend#1\endcsname\endlinenomath}}%
  \newcommand*\patchBothAmsMathEnvironmentsForLineno[1]{%
    \patchAmsMathEnvironmentForLineno{#1}}%
  \patchBothAmsMathEnvironmentsForLineno{equation}%
  \patchBothAmsMathEnvironmentsForLineno{equation*}%
  \patchBothAmsMathEnvironmentsForLineno{align}%
  \patchBothAmsMathEnvironmentsForLineno{align*}%
  \patchBothAmsMathEnvironmentsForLineno{alignat}%
  \patchBothAmsMathEnvironmentsForLineno{alignat*}
  \patchBothAmsMathEnvironmentsForLineno{multline}%
  \patchBothAmsMathEnvironmentsForLineno{multline*}%
  \patchBothAmsMathEnvironmentsForLineno{subequations}%
}

\titlelabel{\thetitle.\space}

\numberwithin{equation}{section}

\newcounter{and}
\newcommand{\institute}[1]{\newcommand{\@institute}{#1}}
\newcommand{\inst}[1]{\unskip\smash{$^#1$}\setcounter{and}{1}\ignorespaces}
\newcommand{\email}[1]{\href{mailto:#1}{#1}}

\renewcommand{\maketitle}{
  { % title
    \raggedright
    \LARGE
    \noindent
    \bfseries
    \maybesf
    \@title
    \par
  }

  \vspace{1.5\baselineskip}

  { % author
    \raggedright
    \renewcommand{\and}{
      \unskip, \ignorespaces
    }
    \noindent\ignorespaces\@author\par
  }

  \vspace{0.5\baselineskip}

  { % institute
    \raggedright
    \small
    \renewcommand{\and}{
      \par\stepcounter{and}
      \hangindent .8em\noindent
      \hbox to .8em{\smash{$^{\arabic{and}}$}}\ignorespaces
    }
    \ifnum\value{and}=0
      \noindent
    \else
      \hangindent .8em\noindent
      \hbox to .8em{\smash{$^{\arabic{and}}$}}\ignorespaces
    \fi
    \ignorespaces\@institute\par
  }
}

\renewenvironment{abstract}{
  \addvspace{1.5\baselineskip}
  \topsep=0pt\partopsep=0pt
  \trivlist\item[\hskip\labelsep\bfseries\maybesf Abstract.]
}{}

\newenvironment{acknowledgments}{
  \addvspace{1.5\baselineskip}
  \topsep=0pt\partopsep=0pt
  \trivlist\item[\hskip\labelsep\bfseries\maybesf Acknowledgments.]
}{}

\newcommand{\latinabbr}{\textit}
\newcommand{\cf}{\latinabbr{cf.}\ }
\newcommand{\eg}{\latinabbr{e.g.}\ }

\newcommand{\ie}{\latinabbr{i.e.}}

\renewcommand{\th}{\ensuremath{^{\text{th}}}\ }
\newcommand{\st}{\ensuremath{^{\text{st}}}\ }
\newcommand{\nd}{\ensuremath{^{\text{nd}}}\ }

\newcommand{\kk}{{\vec{k}}}

\newcommand{\defn}{\doteq} % \equiv or \doteq or \coloneqq or ?

\newcommand{\e}{\mathrm{e}}
\newcommand{\im}{\mathrm{i}}

\newcommand{\field}[1][K]{\mathbb{#1}}
\newcommand{\NN}{\field[N]}
\newcommand{\RR}{\field[R]}

\renewcommand{\Re}{\operatorname{Re}}

\newcommand{\conj}[1]{\overline{#1}}
\DeclarePairedDelimiter{\abs}{\lvert}{\rvert}
\DeclarePairedDelimiter{\norm}{\lVert}{\rVert}
\DeclarePairedDelimiter{\opnorm}{\lVert}{\rVert_{\mathrm{op}}}
\newcommand{\order}[2][]{\mathrm{O}#1( #2 #1)}
\DeclarePairedDelimiter{\E}{\langle}{\rangle} % Expectation value
\newcommand{\norder}[1]{\mathord{:}{#1}\mathord{:}} % Normal order

\newcommand{\dif}{\mathrm{d}}
\newcommand{\gdif}{\mathrm{d}} % Gâteaux derivative
\newcommand{\fdif}{\mathrm{D}} % Fréchet derivative
\newcommand{\od}[3][]{\frac{\dif^{#1}#2}{\dif#3^{#1}}}

\newtheoremstyle{theoremsf}{}{}{\itshape}{}{\bfseries\maybesf}{.}{\labelsep}{}
\newtheoremstyle{definitionsf}{}{}{}{}{\bfseries\maybesf}{.}{\labelsep}{}

\theoremstyle{theoremsf}
\newtheorem{theorem}{Theorem}[section]
\newtheorem{lemma}[theorem]{Lemma}
\newtheorem{proposition}[theorem]{Proposition}

\theoremstyle{definitionsf}
\newtheorem{definition}{Definition}[section]

\definecolor{hypercolor}{rgb}{0,0.2,0.7}
\hypersetup{
  linkcolor=hypercolor,
  urlcolor=hypercolor,
  citecolor=hypercolor,
  pdfauthor=Daniel Siemssen,
  pdfcreator=LaTeX,
  pdfproducer=pdfTeX
}

\newcommand{\Hadamard}{\mathcal{H}}
\newcommand{\KG}{\mathrm{P}}            % Klein--Gordon operator
\newcommand{\Newton}{\mathrm{G}}        % Newton constant

\def\noqed{\renewcommand{\qedsymbol}{}}

\let\emph\textbf

\allowdisplaybreaks

\makeatother

\begin{document}

\title{Global Existence of Solutions of the Semiclassical Einstein Equation for Cosmological Spacetimes}

\author{
  Nicola Pinamonti\inst{1}\inst{2}
	\and
  Daniel Siemssen\inst{1}
}

\institute{
  Dipartimento di Matematica, Università di Genova, Via Dodecaneso 35,
  16146 Genova, Italy.
	\and
  INFN Sezione di Genova, Via Dodecaneso 33,
  16146 Genova, Italy.
	\\
  E-Mail: \email{pinamont@dima.unige.it}, \email{siemssen@dima.unige.it}
}

\maketitle

%----------------------------------------------------------------------------%

\begin{abstract}
  We study the solutions of the semiclassical Einstein equation in flat cosmological spacetimes driven by a massive conformally coupled scalar field.
  In particular, we show that it is possible to give initial conditions at finite time to get a state for the quantum field which gives finite expectation values for the stress-energy tensor.
  Furthermore, it is possible to control this expectation value by means of a global estimate on regular cosmological spacetimes.
  The obtained estimates permit to write a theorem about the existence and uniqueness of the local solutions encompassing both the spacetime metric and the matter field simultaneously.
  Finally, we show that one can always extend local solutions up to a point where the scale factor $a$ becomes singular or the Hubble function $H$ reaches a critical value $H_c = 180 \uppi / \Newton$, which both correspond to a divergence of the scalar curvature $R$, namely a spacetime singularity.
\end{abstract}

%----------------------------------------------------------------------------%

\section{Introduction}

In this paper we shall prove the existence of a unique solution to the \emph{semiclassical Einstein equation} (choosing $8 \uppi\, \Newton = c = \hbar = 1$) with cosmological constant $\Lambda$
\begin{equation}\label{eq:semiclassical}
  G_{ab} + \Lambda g_{ab}= \E{T_{ab}}_\omega
\end{equation}
in the class $\mathcal{O}$ of \emph{homogeneous and isotropic spacetimes} which are also spatially flat and where the matter is described by a massive quantum scalar field conformally coupled with the curvature.
In \cite{pinamonti:2011} one of the authors already showed the local existence for short time-intervals of solutions to \eqref{eq:semiclassical} in Friedman--Lemaître--Robertson--Walker (FLRW) spacetimes with a lightlike initial surface.
In the present paper we shall instead put initial values for matter and gravity at \emph{finite} time, thus describing a more physical and general situation, and prove local and global existence as well as uniqueness of solutions of \eqref{eq:semiclassical}.
However, as we will see shortly, there is a price to pay: the resulting matter state will, in general, not have the full regularity of a Hadamard state.

A systematic analysis of the semiclassical Einstein equation in connection with the problem of regularizing the stress-energy tensor has been given by Wald in \cite{wald:1977}.
Recently, Eltzner and Gottschalk analyzed in \cite{eltzner:2011} the problem of backreaction for scalar quantum fields with generic coupling showing that it gives rise to a dynamical problem which, in principle, can be analyzed numerically.
Similar problems as those discussed in the present paper have been studied by Anderson using methods proper of numerical analysis in a series of papers~\cite{anderson:1983,anderson:1984,anderson:1985,anderson:1986}.

\subsection{Friedman--Lemaître--Robertson--Walker spacetimes}

An element $(M,g)$ of $\mathcal{O}$ is described by the manifold $M = I \times \RR^3$, where $I \subset \RR$ is a connected open interval, and $M$ is equipped with the well-known FLRW metric tensor
\begin{equation*}
  g = - \dif t \otimes \dif t + a(t)^2\, \dif x^i \otimes \dif x^i,
\end{equation*}
where $t$ is the \emph{cosmological time} and $0 < a \in C^2(I)$ is the \emph{scale factor}, the single dynamical degree of freedom.
A very useful feature of this class is the fact that all the elements of $\mathcal{O}$ are conformally flat.
Indeed, for every $(M, g)$ in $\mathcal{O}$, if we use the \emph{conformal time}
\begin{equation}\label{eq:conformal-time}
  \tau \defn \tau_0 - \int_{t_0}^t \frac{1}{a(t')}\, \dif t'
\end{equation}
to parametrize the points of $M$, then $a$ becomes the conformal factor of a conformal transformation which connects $g$ with the flat Minkowski metric $\eta$: $g = a^2 \eta$.
As is clear from the definition, the conformal time is defined up to the choice of $\tau_0$.
Once $t_0$ and $\tau_0$ are fixed and $a$ is known, we can always invert the relation \eqref{eq:conformal-time} to obtain $t$ as a function of $\tau$.
We shall assume this point of view and characterize elements of $\mathcal{O}$ by their scale factor $a(\tau)$ as a function of conformal time.

Since the metrics of the elements of $\mathcal{O}$ have only one degree of freedom, the semiclassical Einstein equations \eqref{eq:semiclassical} can be written in Friedmann form:\footnote{We denote by $\dot f$ the derivative of a function $f$ with respect to cosmological time $t$ and by $f'$ the derivative with respect to conformal time $\tau$.}
\begin{equation*}
  \rho =  3 H^2 - \Lambda,
  \qquad
  p =  - 2 \dot H - 3 H^2 + \Lambda,
\end{equation*}
where $\rho = -{T_0}^0$ is the \emph{energy density}, $p = {T_i}^i/3$ the \emph{pressure} and $H = \dot a / a$ the \emph{Hubble function}.
If the stress-energy tensor is conserved, the previous equations are equivalent to
\begin{equation}\label{eq:friedmann-trace}
  T \defn g_{ab} T^{ab} = - \rho + 3 p = - 6 (\dot H + 2 H^2) + 4 \Lambda
\end{equation}
up to an initial condition which can be fixed by means of the first Friedmann equation $- \rho(\tau_0) = - 3 H(\tau_0)^2 + \Lambda$ evaluated at the initial time $\tau = \tau_0$.

\subsection{Quantum scalar fields}

We are interested in analyzing the role of the quantum nature of $T_{ab}$, to this end we shall consider it as being the stress-energy tensor of a quantized, conformally coupled, massive scalar field with the equation of motion
\begin{equation}\label{eq:eom}
  - \Box \varphi + \frac{1}{6} R \varphi + m^2 \varphi \defn \KG \varphi = 0.
\end{equation}
We recall to the reader that on smooth globally hyperbolic spacetimes there exist unique advanced and retarded fundamental solutions to this equation \cite{bar:2007}.
However, we cannot guarantee the smoothness of the scale factor (and thus the metric) and indeed we will require it only to be $C^1$.
Thanks to the symmetry of the spacetimes in $\mathcal{O}$, the lack of smoothness is not an issue.
In fact, using the Fourier transform, one can construct the unique advanced and retarded fundamental solution explicitly in this case.

The quantization of this scalar field can be performed constructing the algebra generated by the quantum field $\varphi$ \cite{haag:1996,brunetti:2003,hollands:2001}.
In particular, on every smooth $(M, g) \in \mathcal{O}$ we construct the algebra of quantum fields $\mathcal{A}(M, g)$ as the $*$-algebra generated by the formal linear fields
\begin{equation*}
  \varphi(\KG f) = 0,
  \qquad
  \varphi(f)^* = \varphi(\conj{f}),
  \qquad
  [\varphi(f),\varphi(h)] = \im \upDelta(f,h)
\end{equation*}
where $f, h \in \mathcal{D}(M)$ and $\upDelta$ is the unique causal propagator on $(M, g)$, defined as the advanced minus the retarded fundamental solution of \eqref{eq:eom}.
Hence this quantization procedure
\begin{equation*}
  (M, g) \mapsto \mathcal{A}(M, g)
\end{equation*}
is functorial \cite{brunetti:2003}.
The extension of this construction to less regular spacetimes in $\mathcal{O}$ is possible provided a suitable causal propagator exists.

In this picture the expectation values of the fields can be obtained once a state is fixed.
In the algebraic language, a state is a positive, normalized, linear functional over $\mathcal{A}(M, g)$.
Unfortunately, there is no preferred way to select a quantum state in a generic curved space. It is thus not possible to associate a quantum state to the algebra $\mathcal{A}(M, g)$ naturally \cite{fewster:2012}.
However, since in this paper we are interested only in a very special kind of spacetime, we can find a way to associate a special quantum state to every element of $\mathcal{O}$.
More precisely, we will choose a state which resembles the properties of a vacuum at an initial time $\tau_0$.

For physical quasi-free states on smooth globally hyperbolic spacetimes, one often requires that they satisfy the \emph{Hadamard condition} (\cf~\cite{wald:1977}), \ie, in a convex geodesic neighborhood their two-point distribution satisfies
\begin{equation}\label{eq:hadamard_state}
  \omega_2(x,y) = \Hadamard(x, y) + W(x,y) = \lim_{\varepsilon \to 0^+} \frac{U(x,y)}{\sigma_{\varepsilon}(x,y)} + V(x,y) \ln \left( \frac{\sigma_{\varepsilon}(x,y)}{\lambda^2} \right) + W(x,y),
\end{equation}
where $\sigma_{\varepsilon}$ is the $\varepsilon$-regularized, squared, signed geodesic distance between the points $x$ and $y$.
Furthermore, $U$ and $V$ are smooth functions on $M \times M$ determined only by the metric and the equation of motion~\eqref{eq:eom}, and $\Hadamard(x,y)$ is called \emph{Hadamard singularity}.
The smooth function $W$ defined on $M \times M$ characterizes the state and must be chosen such that $\omega_2$ is a positive bisolution of \eqref{eq:eom}.
A lot of progress on the understanding of Hadamard states has been made since it was shown that the Hadamard condition is equivalent to a condition on the wavefront set of $\omega_2$: It must satisfy the \emph{microlocal spectrum condition}~\cite{brunetti:1996,radzikowski:1996,sahlmann:2001}.

Yet, in this paper we have to relax the requirement to consider only Hadamard states.
This will be necessary for two reasons: On one hand, the elements of $\mathcal{O}$ are not smooth and, on the other hand, the quantum states we shall employ are not Hadamard, even on smooth spacetimes.
If one also admits non-smooth $W$ in \eqref{eq:hadamard_state}, the microlocal spectrum condition will fail to hold but can be generalized to a condition on the Sobolev wavefront sets of $\omega_2$ \cite{junker:2002}.
The resulting states are called \emph{adiabatic states} \cite{parker:1969,luders:1990,junker:2002}.
For adiabatic states of order zero $W$ is only a continuous function.

Following \cite{brunetti:2000,hollands:2001}, we can then regularize the two-point distribution of a Hadamard state (or an adiabatic state) just considering
\begin{equation*}
  \omega(\norder{\varphi(x)\varphi(y)}) \defn \omega_2(x,y) - \Hadamard(x, y) = W(x,y).
\end{equation*}
In this manner we can construct the expectation value of the Wick square $\varphi^2$ or of the stress-energy tensor $T_{ab}$ just applying a suitable operator on the previous expression and then taking the coinciding point limit.
This procedure is called \emph{point-splitting regularization}, see~\cite{bunch:1978,moretti:2003,hollands:2001} for further details.

However, henceforth we consider spacetimes in the set $\mathcal{O}$ which have in general no smooth metric.
For non-smooth metrics the Hadamard construction fails but due to the high symmetry of the spacetime it is still possible to construct states which resemble adiabatic states of order $0$.
Moreover, one can introduce a regularization procedure which coincides with the usual point-splitting method for adiabatic states on smooth spacetimes.
For a detailed discussion see section~\ref{sec:state}.

\subsection{The semiclassical Einstein equation}

Having reviewed these standard results, we are now ready to rewrite the semiclassical Einstein equation \eqref{eq:semiclassical} for all \emph{smooth} spacetimes in $\mathcal{O}$ as a semiclassical version of \eqref{eq:friedmann-trace}:
\begin{equation}\label{eq:friedmann-trace-semiclassical}
  - 6 (\dot H + 2 H^2) + 4 \Lambda = \E{T}_\omega.
\end{equation}
The trace of the quantum stress-energy tensor for a conformally coupled scalar field is (see~\cite{dappiaggi:2008,moretti:2003} for further details):
\begin{equation}\label{eq:trace_T}
  \E{T}_\omega = \omega(\norder{T}) + T_{\mathrm{RF}} = - m^2 \omega(\norder{\varphi^2}) + 2 [V_1] - \alpha\, m^4 - \beta\, m^2 R - \gamma\, \Box R,
\end{equation}
where $\omega(\norder{\varphi^2}) = [W]$ is due to the state-dependent part of the state (\cf~\eqref{eq:hadamard_state}) and the contribution
\begin{align*}
  2 [V_1] & = \frac{1}{8 \uppi^2} \left( -\frac{H^2}{30} (\dot H + H^2) + \frac{1}{360} \Box R + \frac{m^4}{4} \right)
\end{align*}
is usually called \emph{trace anomaly} \cite{wald:1978}.
Moreover, $[\,\cdot\,]$ denotes the Synge bracket of a bitensor: $[B(x,y)] \defn \lim_{y \to x} B(x,y)$.

The constants $\alpha, \beta$ and $\gamma$ appearing in \eqref{eq:trace_T} are renormalization constants which come from the renormalization freedom $T_{RF}$ to choose a state-independent, local normal ordering prescription $\norder{\,\cdot\,}$.
According to \cite{brunetti:2003,hollands:2001}, they have to be fixed once and for all for every element of $\mathcal{A}(M)$ with $M \in \mathcal{O}$ in accordance with experiments.
In order to describe solutions of the semiclassical equation~\eqref{eq:friedmann-trace-semiclassical}, we shall fix them according to the following rules:

We will choose $\gamma$ in such a way as to cancel higher order derivatives of the metric coming from the trace anomaly.
Following \cite{wald:1977}, this is necessary because we want to have a well-posed initial value problem without adding extra initial conditions.
Removing this term might not be suitable for describing the physics close to the initial Big Bang singularity.\footnote{In the Starobinsky model this term is the single term which is considered to drive a phase of rapid expansion close to the Big Bang, see the original paper of Starobinsky \cite{starobinsky:1980}, its further development \cite{kofman:1985} and also \cite{hack:2010,hack:2013} for a recent analysis.}
However, this is surely suitable to describe the physics in the regime where $\Box H \ll H^4$.

Regarding the physical meaning of the remaining two renormalization constants, we notice that changing $\alpha$ corresponds to a renormalization of the cosmological constant $\Lambda$, whereas a change of $\beta$ corresponds to a renormalization of the Newton constant $\Newton$.
For this reason we choose $\alpha$ in such a way that no contribution proportional to $m^4$ is present in $\E{T}_\omega$ and we set $\beta$ in order to cancel the terms proportional to $m^2 R$ in $\E{T}_\omega$ because this contribution is already fixed by the measured value of $\Newton$.
All in all, we choose the renormalization constants as\footnote{The value of $\beta$ is not zero because, as it will be clear later (\cf section~\ref{sub:wick_square} and in particular \eqref{eq:adiabatic_subtraction}, \eqref{eq:renormalized_wick_square}), there is a contribution proportional to $R$ in the expectation value of $\varphi^2$ for the chosen class of states.}
\begin{equation}\label{eq:renormalization}
  \alpha = \frac{1}{32 \uppi^2},
  \qquad
  \beta = \frac{1}{288 \uppi^2},
  \qquad
  \gamma = \frac{1}{2880 \uppi^2}.
\end{equation}

With all this in mind, we can write the semiclassical equation \eqref{eq:friedmann-trace-semiclassical}, in the following integral form
\begin{multline}\label{eq:functional-volterra}
  H(\tau) = H_0 + \int_{\tau_0}^\tau \frac{a(H)(\eta)}{H^2_c - H(\eta)^2} \Big( H(\eta)^4 - 2 H_c^2 H(\eta)^2 \\ + 240 \uppi^2 \big(m^2 \omega(\norder{\varphi^2})(H)(\eta) + \beta\, m^2 R(H)(\eta) + 4 \tilde\Lambda \big) \Big)\, \dif\eta,
\end{multline}
where $H_c^2 \defn 1440 \uppi^2 / (8 \uppi\, \Newton) = 180 \uppi / \Newton$ and $\tilde\Lambda \defn \Lambda / (8 \uppi\, \Newton)$.
Moreover, the initial condition $H_0 \defn H(\tau_0)$ is fixed (up to a sign) by the constraint
\begin{equation}\label{eq:constraint}
  H_0^2 = \frac{1}{3} \big(\rho(\tau_0)+\Lambda \big).
\end{equation}
The equation \eqref{eq:functional-volterra} can be rewritten as
\begin{equation}\label{eq:functional-volterra-short}
  H = F(H)
  \quad \text{with} \quad
  F(H)(\tau) \defn H_0 + \int_{\tau_0}^\tau f(H)(\eta)\, \dif\eta,
\end{equation}
where $f$ is a suitable functional of $H$ which also depends on $\tau_0$ through the state $\omega$.
This expression exhibits the structure of a Volterra integral equation where the integral kernel is an integro-differential operator.

Given a \emph{non-smooth} spacetime in $\mathcal{O}$, \eqref{eq:functional-volterra} is still well-defined if we have a suitable notion of the Wick square $\omega(\norder{\varphi^2})$.
Therefore, in the \hyperref[sec:state]{next section}, we will present the detailed construction of the states we are considering and we will discuss the well-posedness of the Wick square as a functional of $H$.
Then, in order to construct the functional $f$ explicitly, we have to choose a state for every element of $\mathcal{O}$ and analyze how the resulting Wick square $\omega(\norder{\varphi^2})$ in $f$ depends on $H$.

Afterwards, in the \hyperref[sec:local]{third section}, we will discuss the existence and uniqueness of solutions to \eqref{eq:functional-volterra}.
This problem amounts to finding fixed points of the map on the right hand side of \eqref{eq:functional-volterra}.
Using the Banach fixed-point theorem, this will be accomplished by showing that $F$ is a contraction map on a suitable closed subset of a Banach space.
To prove that $F$ is a contraction map we will use results derived in the \hyperref[sec:appendix]{Appendix} which essentially reduce the problem to showing that the first functional derivative of $\omega(\norder{\varphi^2})$ in $f$ can be controlled.

In the \hyperref[sec:global]{fourth section} we shall show that it is always possible extend the (local) solution provided that the functional $f$ remains bounded, thus proving global existence and uniqueness.
This result descends from the estimates for the Wick square given in the preceding sections.

We conclude this paper with some final comments in the \hyperref[sec:final]{fifth section}.

%----------------------------------------------------------------------------%

\section{Quantum states for the scalar field}
\label{sec:state}

As described above, in the present paper we would like to solve the semiclassical Einstein equation giving initial conditions at finite conformal time $\tau=\tau_0$.
For a coherent picture, we should give initial values for the state describing quantum matter at the same time using information about the metric and its first derivative.

It would be desirable to use Hadamard states as reference states.
In the literature there are a few examples of such states but unfortunately none of them are suitable for our purposes.
On FLRW spacetimes there is the notable construction of states of low energy given by Olberman~\cite{olbermann:2007} which are also Hadamard, see also~\cite{kusku:2008,hack:2013}.
But the employed construction is based on a smearing of the modes with respect to an extended function of time and, in principle, we do not even know if a spacetime, namely a solution of \eqref{eq:functional-volterra}, exists in the future of $\tau_0$.
Other constructions of Hadamard states require that the spacetime has certain asymptotic properties \cite{dappiaggi:2009a,dappiaggi:2009b,pinamonti:2011} which are not under control in the class $\mathcal{O}$.

The price we have to pay for working with non-smooth spacetimes is the reduced regularity of the obtained state.
More precisely, we cannot say that it is a Hadamard state and at most we will obtain a state which is as close as possible to an adiabatic state of order zero.
The construction of adiabatic states can be seen in a paper by Lüders and Roberts~\cite{luders:1990}, making precise previous ideas of Parker~\cite{parker:1969}.
The relation of the adiabatic construction with Hadamard property and the microlocal spectrum condition was later analyzed by Junker and Schrohe~\cite{junker:2002}.

\subsection{Construction}
\label{sub:state_construction}

We shall start our discussion by outlining the construction of the states that we consider in this paper:
We take as states of $\varphi$ on $(M, g) \in \mathcal{O}$ the pure homogeneous quasi-free states $\omega$ whose two-point distribution is given by
\begin{equation}\label{eq:two_point}
  \omega_2(x, y)
  = \lim_{\varepsilon \to 0^+} \frac{1}{(2 \uppi)^3} \int_{\RR^3} \frac{\conj\chi_k(\tau_x)}{a(\tau_x)} \frac{\chi{}_k(\tau_y)}{a(\tau_y)}\, \e^{\im \vec{k} \cdot (\vec{x} - \vec{y})} \e^{-\varepsilon k}\, \dif\vec{k}
\end{equation}
with modes $\chi_k$ satisfying the mode equation
\begin{equation}
  \label{eq:eom_modes}
  \chi_k''(\tau) + \big( k^2 + a(\tau)^2 m^2 \big)\, \chi_k^{\phantom{'}}(\tau) = 0
\end{equation}
and initial conditions given by the zeroth order adiabatic mode, namely\footnote{Notice that the commutator condition $\chi_k' \conj\chi{}_k^{\phantom{'}} - \chi_k^{\phantom{'}} \conj\chi{}_k' = \im$ holds everywhere if it holds on an arbitrary Cauchy surface. Given the chosen initial conditions, it evidently holds for $\tau = \tau_0$.}
\begin{equation}
  \label{eq:initial_conditions_modes}
  \chi_k^{\phantom{'}}(\tau_0) = \frac{1}{\sqrt{2 k_0}}\, \e^{\im k_0 \tau_0},
  \quad
  \chi_k'(\tau_0) = \frac{\im k_0}{\sqrt{2 k_0}}\, \e^{\im k_0 \tau_0},
\end{equation}
where $k_0^2 \defn k^2 + a_0^2 m^2, a_0 \defn a(\tau_0)$.
Note that these states are adiabatic states of order zero \cite{junker:2002,luders:1990} whenever the scale factor $a$ is smooth.

We can solve \eqref{eq:eom_modes} perturbatively around the initial conditions specified above.
That is, introducing the perturbation potential $V(\tau) = m^2 ( a(\tau)^2 - a_0^2 )$, we obtain a solution $\chi_k(\tau) = \sum_n \chi_k^n(\tau)$ by solving the recurrence relation
\begin{equation}
  \label{eq:recurrence_modes}
  \chi_k^n{}''(\tau) + k_0^2\, \chi_k^n(\tau) = - V(\tau)\, \chi_k^{n-1}(\tau)
\end{equation}
for $n > 0$ with initial condition $\chi_k^0(\tau) = (2 k_0)^{-1/2} \exp(\im k_0 \tau)$.
For $\tau > \tau_0$ this is accomplished by applying the retarded propagator of $\partial_\tau^2 + k_0^2$ to \eqref{eq:recurrence_modes}, \ie,
\begin{equation}
  \label{eq:inverted_recurrence}
  \chi_k^n(\tau) = \int_{\tau_0}^\tau \frac{\sin\big( k_0 (\eta - \tau) \big)}{k_0}\, V(\eta)\, \chi_k^{n-1}(\eta)\, \dif\eta.
\end{equation}
The convergence of this Ansatz can be easily shown, \cf~\cite[proposition 4.4]{pinamonti:2011}, because
\begin{equation}
  \label{eq:estimate_chi_n}
  \abs{\chi_k^n} \leq \frac{1}{\sqrt{2 k_0}\, n!}\, \left( \frac{1}{k_0} \int_{\tau_0}^\tau \abs{V(\eta)}\, \dif\eta \right)^l \left( \int_{\tau_0}^\tau (\tau - \eta)\, \abs{V(\eta)}\, \dif\eta \right)^{n-l}
\end{equation}
for any $0 \leq l \leq n$.
We further note that while the modes $\chi_k$ are well-defined for scale factors $a$ which are less regular than $C^2$, in that case the scalar curvature $R$ and the equation of motion~\eqref{eq:eom} are ill-defined, hence the mode equation~\eqref{eq:eom_modes} cannot be consistently derived from \eqref{eq:eom} and the modes do not define a proper state via~\eqref{eq:two_point}.

\subsection{Wick square}
\label{sub:wick_square}

The equation~\eqref{eq:functional-volterra} that we seek to solve contains the Wick square of $\varphi$ in a state $\omega$ and thus (on a smooth spacetime) we need to compute
\begin{equation}\label{eq:wick_square}
  \omega\big( \norder{\varphi^2}(x) \big) = \lim_{y \to x} \big( \omega_2(x, y) - \Hadamard(x, y) \big).
\end{equation}
Instead of performing the minimal subtraction $\omega_2 - \Hadamard$ directly, we follow an approach analogous to that in \cite{pinamonti:2011} to perform an equivalent subtraction on the level of modes.
First of all, a direct calculation shows that for $x, y$ at equal time $\tau$
\begin{multline}\label{eq:adiabatic_subtraction}
  \Hadamard\big( (\tau, \vec{x}), (\tau, \vec{y}) \big) - \lim_{\varepsilon \to 0^+} \frac{1}{(2 \uppi)^3 a(\tau)^2} \int_{\RR^3} \left( \frac{1}{2 k_0} - \frac{V(\tau)}{4 k_0^3} \right) \e^{\im \vec{k} \cdot (\vec{x} - \vec{y})} \e^{-\varepsilon k}\, \dif\vec{k} \\
  = \frac{m^2}{16 \uppi^2} \left( \left(\frac{a_0}{a(\tau)}\right)^2 - 2 \ln\left( \frac{a_0}{a(\tau)} \right) - 2 \ln\left(\frac{\e^{\gamma} m\, \lambda}{\sqrt{2}}\right) \right) + \frac{R(\tau)}{288 \uppi^2} + \order[\big]{\abs{\vec{x} - \vec{y}}^2},
\end{multline}
where $\gamma$ is the Euler--Mascheroni constant and $\lambda$ is the length scale of the Hadamard parametrix (\cf~\cite{moretti:2003}).
Note that the curvature term $R/(288 \uppi^2)$ is exactly canceled by our choice of the renormalization constant $\beta$ in \eqref{eq:renormalization}.
Therefore, subtracting the Hadamard parametrix is (up to the terms indicated above and a conformal rescaling) equivalent to subtracting $1/(2 k_0) - V/(4 k_0^3)$ in Fourier space.
Neglecting the terms on the right-hand side of \eqref{eq:adiabatic_subtraction} for now, this subtraction is still well-defined if the scale factor is not smooth and we can show that it gives indeed a finite result in the coinciding point limit if $\omega_2$ is given by~\eqref{eq:two_point}:

\begin{proposition}\label{prop:regularized_state}
  The regularized two-point distribution
  \begin{equation*}
    \omega_2\big( (\tau, \vec{x}), (\tau, \vec{y}) \big) - \lim_{\varepsilon \to 0^+} \frac{1}{(2 \uppi)^3 a(\tau)^2} \int_{\RR^3} \left( \frac{1}{2 k_0} - \frac{V(\tau)}{4 k_0^3} \right) \e^{\im \vec{k} \cdot (\vec{x} - \vec{y})} \e^{-\varepsilon k}\, \dif\vec{k},
  \end{equation*}
  with $\omega_2$ given by \eqref{eq:two_point}, converges in the coinciding point limit.
\end{proposition}

\begin{proof}
  We have to show that
  \begin{equation}\label{eq:regularized_state}
    \lim_{\varepsilon \to 0^+} \int_{\RR^3} \left( \abs{\chi_k(\tau)}^2 - \frac{1}{2 k_0} + \frac{V(\tau)}{4 k_0^3} \right) \e^{-\varepsilon k}\, \dif\vec{k}
    = \int_{\RR^3} \left( \abs{\chi_k(\tau)}^2 - \frac{1}{2 k_0} + \frac{V(\tau)}{4 k_0^3} \right)\, \dif\vec{k}
  \end{equation}
  is finite.
  To this end we expand the product $\abs{\chi_k}^2$ with $\chi_k = \sum_n \chi_k^n$ as
  \begin{equation*}
    \abs{\chi_k}^2 = \sum_{n=0}^\infty \sum_{l=0}^n \chi_k^l\, \conj\chi{}_k^{n-l}
  \end{equation*}
  in terms of the order $n$.
  Inserting this expansion into \eqref{eq:regularized_state}, we can prove the thesis order by order:

  \begin{proof}[0\th Order]
    Since
    \begin{equation*}
      \chi_k^0\, \conj\chi{}_k^0 = \frac{1}{2 k_0},
    \end{equation*}
    the first term $(2 k_0)^{-1}$ in the subtraction exactly cancels the zeroth order term $\abs{\chi_k^0}^2$ in~\eqref{eq:regularized_state}.
    \noqed
  \end{proof}

  \begin{proof}[1\st Order]
    Using an integration by parts, we can rewrite the first order terms as
    \begin{align*}
      (\chi_k^0\, \conj\chi{}_k^1 + \chi_k^1\, \conj\chi{}_k^0)(\tau)
      & = \frac{1}{k_0^2} \int_{\tau_0}^\tau \sin\big( k_0 (\eta - \tau) \big) \cos\big( k_0 (\eta - \tau) \big) V(\eta)\, \dif\eta \\
      & = \frac{1}{2 k_0^2} \int_{\tau_0}^\tau \sin\big( 2 k_0 (\eta - \tau) \big) V(\eta)\, \dif\eta \\
      & = - \frac{V(\tau)}{4 k_0^3} + \frac{1}{4 k_0^3} \int_{\tau_0}^\tau \cos\big( 2 k_0 (\eta - \tau) \big) V'(\eta)\, \dif\eta.
    \end{align*}
    While the first summand $V(\tau) (4 k_0^3)^{-1}$ in the last line is exactly canceled by the second term in the subtraction in \eqref{eq:regularized_state}, the second summand yields
    \begin{align}
      \notag
      \MoveEqLeft \int_{\RR^3} \frac{1}{4 k_0^3} \left( \int_{\tau_0}^\tau \cos\big( 2 k_0 (\eta - \tau) \big) V'(\eta)\, \dif\eta \right) \e^{-\varepsilon k}\, \dif\kk \\
      \notag
      & = \uppi \int_0^\infty \frac{k^2}{k_0^3} \left( \int_{\tau_0}^\tau \cos\big( 2 k_0 (\eta - \tau) \big) V'(\eta)\, \dif\eta \right) \e^{-\varepsilon k}\, \dif k \\
      \notag
      & = \uppi \int_{a_0 m}^\infty k_0^{-1} \sqrt{1 - a^2 k_0^{-2}}\, \left( \int_{\tau_0}^\tau \cos\big( 2 k_0 (\eta - \tau) \big) V'(\eta)\, \dif\eta \right) \e^{-\varepsilon k}\, \dif k_0 \\
      \label{eq:1st_order_k_integral}
      & = \uppi \int_{\tau_0}^\tau V'(\eta)\, \left( \int_{a_0 m}^\infty k_0^{-1} \cos\big( 2 k_0 (\eta - \tau) \big)\, \e^{-\varepsilon k}\, \dif k_0 \right)\, \dif\eta - R(\tau)
    \end{align}
    for $\varepsilon > 0$.
    Here $R$ is a finite remainder term since it contains terms in the $k_0$-integration which decay at least like $k_0^{-3}$.
    Notice that, in the last equation of the previous formula, thanks to the positivity of $\varepsilon$ we have switched the order in which the $k_0$- and $\eta$-integration are taken.
    We would like to show that, in the last expression of \eqref{eq:1st_order_k_integral}, the weak limit $\varepsilon\to 0^+$ can be taken before the $\eta$-integration.

    To this end it remains to be shown that the $k_0$-integral in the first summand of \eqref{eq:1st_order_k_integral} converges in the limit $\varepsilon \to 0^+$ to an integrable function in $[\tau_0, \tau]$.
    First, note that the exponential integral
    \begin{equation}\label{eq:incomplete_gamma}
      E_1(x) = \Gamma(0, x)
      = \int_1^\infty \frac{\e^{-x t}}{t}\, \dif t
      = \int_0^1 \frac{\e^{-x}}{x - \ln (1 - s)}\, \dif s
    \end{equation}
    converges for $x \neq 0, \Re x \geq 0$.
    To show the identity, we used the substitution $t = -x^{-1} \ln (1 - s) + 1$ involving a subtle but inconsequential change of the integration contour in the complex plane if $x$ is complex.
    Then we easily see that
    \begin{equation}\label{eq:gamma_k_integral}
      \lim_{\varepsilon \to 0^+} \int_{a_0 m}^\infty k_0^{-1} \e^{\pm 2 \im k_0 (\eta - \tau) - \varepsilon k_0}\, \dif k_0 = E_1\big( \!\pm 2 \im a_0 m (\eta - \tau) \big)
    \end{equation}
    converges for $\eta \neq \tau$.
    This result is related with \eqref{eq:1st_order_k_integral} via
    \begin{equation*}
      \lim_{\varepsilon \to 0^+} \int_{a_0 m}^\infty k_0^{-1} \cos\big( 2 k_0 (\eta - \tau) \big)\, \e^{-\varepsilon k}\, \dif k_0
      = \lim_{\varepsilon \to 0^+} \int_{a_0 m}^\infty k_0^{-1} \cos\big( 2 k_0 (\eta - \tau) \big)\, \e^{-\varepsilon k_0}\, \dif k_0,
    \end{equation*}
    where we used the boundedness of $k - (k^2 - a_0^2\, m^2)^{1/2}$ and \eqref{eq:gamma_k_integral}.
    Finally, a bound sufficient to see the $\eta$-integrability of the $k_0$-integral in \eqref{eq:1st_order_k_integral} can be obtained from the identity in \eqref{eq:incomplete_gamma}, namely,
    \begin{equation*}
      \abs*{E_1(\im x)}
      = \abs*{\int_0^1 \big( \im x - \ln (1 - s) \big)^{-1}\, \dif s}
      \leq \int_0^1 \big( x^2 + s^2 \big)^{-1/2}\, \dif s
      = \ln\bigg( \frac{1 + \sqrt{1 + x^2}}{x} \bigg).
    \end{equation*}
    \noqed
  \end{proof}

  \begin{proof}[2\nd Order]
    For the second order we calculate
    \begin{align}
      \notag
      \MoveEqLeft (\chi_k^0\, \conj\chi{}_k^2 + \chi_k^1\, \conj\chi{}_k^1 + \chi_k^2\, \conj\chi{}_k^0)(\tau) \\
      \notag
      & = \!\begin{multlined}[t]
        \frac{1}{k_0^3} \int_{\tau_0}^\tau \sin\big( k_0 (\eta - \tau) \big) V(\eta)\, \left( \int_{\tau_0}^\eta \sin\big( k_0 (\xi - \eta) \big) V(\xi) \cos\big( k_0 (\xi - \tau) \big)\, \dif\xi \right. \\
        + \left. \frac{1}{2} \int_{\tau_0}^\tau \sin\big( k_0 (\xi - \tau) \big) V(\xi) \cos\big( k_0 (\xi - \eta) \big)\, \dif\xi \right)\, \dif\eta
      \end{multlined} \\
      \notag
      & = \frac{1}{k_0^3} \int_{\tau_0}^\tau \sin\big( k_0 (\eta - \tau) \big) V(\eta)\, \left( \int_{\tau_0}^\eta \sin\big( k_0 (2 \xi - \eta - \tau) \big) V(\xi)\, \dif\xi \right)\, \dif\eta \\
      \label{eq:2nd_order_final} % FIXME (tag placement)
      & = \!\begin{multlined}[t]
        \frac{1}{2 k_0^4} \int_{\tau_0}^\tau \sin\big( k_0 (\eta - \tau) \big) V(\eta)\, \left( \int_{\tau_0}^\eta \cos\big( k_0 (2 \xi - \eta - \tau) \big) V'(\xi)\, \dif\xi \right. \\
        \left.\vphantom{\int_{\tau_0}^\eta} - \cos\big( k_0 (\eta - \tau) \big) V(\eta) \right)\, \dif\eta,
      \end{multlined}
    \end{align}
    where we used integration by parts in the last equality.
    It is easy to obtain a $\kk$-uniform estimate for the integral above and thus the integrability of the second order follows from $\int_{\RR^3} k_0^{-4}\, \dif\kk < \infty$.
    \noqed
  \end{proof}

  \begin{proof}[Higher Orders]
    For orders $n > 2$ it is sufficient to use the rough estimate from~\eqref{eq:estimate_chi_n}:
    \begin{align}
      \notag
      \abs*{\sum_{n=3}^\infty \sum_{l=0}^n \chi_k^l\, \conj\chi{}_k^{n-l}}(\tau)
      & \leq \frac{1}{2 k_0} \sum_{n=3}^\infty \frac{2^n}{n!} \left( \frac{1}{k_0} \int_{\tau_0}^\tau \abs[\big]{V(\eta)}\, \dif\eta \right)^3 \left( \int_{\tau_0}^\tau (\tau - \eta)\, \abs[\big]{V(\eta)}\, \dif\eta \right)^{n-3} \\
      \label{eq:higher_order_final}
      & \leq \frac{4}{k_0^4} \left( \int_{\tau_0}^\tau \abs[\big]{V(\eta)}\, \dif\eta \right)^3 \exp\left( 2 \int_{\tau_0}^\tau (\tau - \eta)\, \abs[\big]{V(\eta)}\, \dif\eta \right).
    \end{align}
    As above, the integrability of the higher orders follows from $\int_{\RR^3} k_0^{-4}\, \dif\kk < \infty$.
    \noqed
  \end{proof}

  \noindent Combining these partial results, we see that the thesis holds true.
\end{proof}

Therefore we can consistently define the renormalized Wick square at conformal time $\tau$ for every element of $\mathcal{O}$ as
\begin{multline}\label{eq:renormalized_wick_square}
  \omega\big(\norder{\varphi^2}(\tau)\big) + \alpha\, m^2 + \beta\, R(\tau) \defn \frac{1}{(2 \uppi)^3 a(\tau)^2} \int_{\RR^3} \left( \abs{\chi_k(\tau)}^2 - \frac{1}{2 k_0} + \frac{V(\tau)}{4 k_0^3} \right)\, \dif\vec{k} \\ + \frac{m^2}{(4 \uppi)^2} \left( \frac{1}{2} - \left(\frac{a_0}{a(\tau)}\right)^2 + 2 \ln\left( \frac{a_0}{a(\tau)} \right) + 2 \ln\left(\frac{\e^{\gamma} m\, \lambda}{\sqrt{2}}\right) \right) ,
\end{multline}
which coincides with \eqref{eq:wick_square} for smooth spacetimes (up to the added renormalization freedom).
Moreover, we notice that, as a consequence of the previous proposition, it is possible to obtain global estimates for the renormalized Wick square:

\begin{lemma}\label{lem:bounds}
  The renormalized Wick square is bounded on every $a' \in C[\tau_0,\tau_1]$ with $a > 0$ in $[\tau_0,\tau_1]$ for every $\tau_1$ and with $a(\tau_0) = a_0$, namely,
  \begin{equation*}
    \abs*{\omega\big( \norder{\varphi^2}(\tau_1) \big) + \alpha\, m^2 + \beta\, R(\tau_1)} \leq C\left(\sup_{[\tau_0,\tau_1]}{a},\sup_{[\tau_0,\tau_1]}{a'}, (\tau_1-\tau_0),\frac{1}{\inf_{[\tau_0,\tau_1]}a } \right) \;
  \end{equation*}
  where $C$ is a finite increasing function.
\end{lemma}
\begin{proof}
  The proof of this lemma and the explicit value of $C$, can be obtained combining \eqref{eq:wick_square} with \eqref{eq:adiabatic_subtraction} and then analyzing the conformal adiabatic subtraction \eqref{eq:regularized_state} order by order as in the proof of the preceding proposition.
\end{proof}

\subsection{Energy density}

There is another nice feature about the states we have constructed above.
Thanks to the conformal coupling with the curvature, the energy density computed in these states is finite even though these states are (on smooth spacetimes) only adiabatic states of order zero.
This is a crucial feature which permits us to solve the constraint~\eqref{eq:constraint}, \ie, $H_0^2 = 1/3\, (\rho(\tau_0) + \Lambda)$, as a first step towards solving the semiclassical Einstein equation.

\begin{proposition}
  The energy density $\rho$ in the state $\omega$ at the initial time $\tau = \tau_0$ is finite.
\end{proposition}
\begin{proof}
  Following \cite{hack:2013}, in order to show that $\rho(\tau_0)$ is finite, we just need to show that the adiabatically regularized expression
  \begin{equation}\label{eq:adiabatic-conformal-energy}
    \int_0^\infty \Big(
      \big( \abs{\chi_k'}^2 + (k^2 + m^2 a^2) \abs{\chi_k}^2 \big)
      -
      \big( \abs{\chi_{k,0}'}^2 + (k^2 + m^2 a^2) \abs{\chi_{k,0}}^2 \big)
    \Big)\, k^2\, \dif k
  \end{equation}
  does not diverge at $\tau = \tau_0$.
  Here $\chi_k$ are the modes constructed in section~\ref{sub:state_construction}, whereas the functions
  \begin{equation*}
    \chi_{k,0} \defn \frac{1}{\sqrt{2} (k^2 + m^2 a^2)^{1/4}} \exp\left( \im \int_{\tau_0}^\tau \sqrt{k^2 + m^2 a^2(\eta)}\, \dif\eta \right)
  \end{equation*}
  are called adiabatic modes of order zero by Parker~\cite{parker:1969}.
  Subtracting Parker's adiabatic modes in \eqref{eq:adiabatic-conformal-energy} is essentially equivalent to the Hadamard regularization for the given choice of renormalization constants \eqref{eq:renormalization}.
  Evaluating the expression~\eqref{eq:adiabatic-conformal-energy} at $\tau = \tau_0$ gives
  \begin{equation*}
    \frac{m^4}{8} \int_0^\infty \frac{a_0^2\, (a')^2}{ (k^2 + m^2 a_0^2)^{5/2}}\, k^2\, \dif k < \infty.
    \qedhere
  \end{equation*}
\end{proof}
Notice that the previous proposition only guarantees that the energy density is well-defined at the initial time.
Nevertheless, the conservation equation for the stress-energy tensor permits to state that it is well-defined everywhere.

The expression~\eqref{eq:adiabatic-conformal-energy} coincides with the energy density $\rho$ of the system up to a conformal rescaling and up to the addition of some finite terms.
Thus, since the energy density $\rho$ is finite in the considered state, the constraint~\eqref{eq:constraint} holds, provided a suitable choice of $H(\tau_0)$ and $\Lambda$ is made.
We stress that, if we do not want to alter $\Lambda$, the same result can be achieved adding classical radiation to the energy density of the universe in a suitable state.

We would like to conclude this section with a remark.
In adiabatic states of order zero the expectation values of local fields containing derivatives are usually ill-defined.
Despite this, in the case of conformal coupling and for our choice of initial conditions~\eqref{eq:initial_conditions_modes}, the energy density turns out to be well-defined.
This is essentially due to the fact that in the massless conformally coupled case the adiabatic modes of order zero are solutions of the mode equation~\eqref{eq:eom_modes} and in that case the obtained state is the well known conformal vacuum. Hence, the adiabatically regularized energy density vanishes.
In the massive case the states constructed above are not very different than the conformal vacuum and, in particular, the energy density remains finite under that perturbation.

%----------------------------------------------------------------------------%

\section{Local solutions}
\label{sec:local}

Our aim is to show the existence and uniqueness of local solutions to the semiclassical Einstein equation in the class $\mathcal{O}$ of FLRW spacetimes.
In particular, according to the discussion in the introduction, we will analyze the uniqueness and existence of solutions of \eqref{eq:functional-volterra}.
Similar to the Picard--Lindelöf theorem, we will use the Banach fixed-point theorem to achieve this goal.
Some results on functional derivatives and the Banach fixed-point theorem are collected in the \hyperref[sec:appendix]{appendix}.

First we have to select a Banach space for candidate Hubble functions $H$ that the functional $F$ (see eqs. \eqref{eq:functional-volterra} and \eqref{eq:functional-volterra-short}) operates upon:
$F$ can be considered as a functional on the Banach spaces\footnote{Until fixed, we take both $\tau_0$ and $\tau_1$ as variable and thus consider a family of Banach spaces.} $C[\tau_0, \tau_1]$, $\tau_0 < \tau_1 $, equipped with the uniform norm
\begin{equation*}
  \norm{X}_{C[\tau_0,\tau_1]} \defn \norm{X}_\infty \defn \sup_{\tau \in [\tau_0,\tau_1]} \abs{X(\tau)}.
\end{equation*}
However, once $\tau_0$ and the initial condition $a_0 = a(H)(\tau_0) > 0$ are fixed, we find that
\begin{equation}\label{eq:a_of_H}
  a(H)(\tau) = a_0 \left(1 - a_0 \int_{\tau_0}^\tau H(\eta)\, \dif\eta \right)^{-1},
\end{equation}
as a functional of $H$, is not continuous on $C[\tau_0, \tau_1]$.
But we can find an open subset
\begin{equation}\label{eq:set_U}
  \mathcal{U}[\tau_0,\tau_1] \defn \Big\{ H \in C[\tau_0, \tau_1] \;\Big|\; \norm{H}_{C[\tau_0,\tau_1]} < \min\big\{ \big(a(H)(\tau_0) (\tau_1 - \tau_0)\big)^{-1}, H_c \big\} \Big\}
\end{equation}
on which $a$ and thus also $V = m^2 (a^2 - a_0^2)$ depend smoothly on $H$.
Indeed, we can show the following:

\begin{proposition}\label{prop:gateaux_diffable}
  The functional
  \begin{equation}\label{eq:functional_f}
    f(H) = \frac{a(H)}{H^2_c - H^2} \Big( H^4 - 2 H_c^2 H^2 + 240 \uppi^2 \big(m^2 \omega(\norder{\varphi^2})(H) + \beta\, m^2 R(H) + 4 \tilde\Lambda \big) \Big)
  \end{equation}
  is continuously Gâteaux differentiable on $\mathcal{U}[\tau_0,\tau_1]$ for arbitrary but fixed $\tau_0, \tau_1$ and $a_0 = a(\tau_0)$.
\end{proposition}
\begin{proof}
  Given \eqref{eq:wick_square}, \eqref{eq:adiabatic_subtraction}, proposition~\ref{prop:regularized_state} and lemma~\ref{lem:bounds}, it is enough to show that $a(H)$ and $(H_c^2 - H^2)^{-1}$ are bounded and that $\omega(\norder{\varphi^2})(\tau_0, H)$ is continuously Gâteaux differentiable.
  The former is assured by the condition $\norm{H}_{C[\tau_0,\tau_1]} < \min\{ a_0^{-1} (\tau_1 - \tau_0)^{-1}, H_c \}$ in the definition of $\mathcal{U}[\tau_0,\tau_1]$.
  For the latter it remains to be shown that the renormalized Wick square \eqref{eq:renormalized_wick_square} is continuously Gâteaux differentiable on $\mathcal{U}[\tau_0,\tau_1]$:

  We start by calculating the functional derivative of the scale factor
  \begin{equation*}
    \gdif a(H, \delta H)(\tau) = a(H)(\tau)^2 \int_{\tau_0}^\tau \delta H(\eta)\, \dif\eta.
  \end{equation*}
  The functional derivatives for $a^{-2}$, $\ln a$, $V$ and $V'$ follow easily.
  In particular we note that all these functions are continuously Gâteaux differentiable on $\mathcal{U}[\tau_0,\tau_1]$ because integration is a continuous operation and $a$ depends smoothly on $H$ in $\mathcal{U}[\tau_0,\tau_1]$.
  Therefore it suffices to analyze the differentiablity of the integral \eqref{eq:regularized_state} appearing in the regularized two-point distribution.
  Moreover, within $\chi_k$ only the potential $V$ is (smoothly on $\mathcal{U}[\tau_0,\tau_1]$) dependent on $H$, thus simplifying the computations considerably.\footnote{If we were to work in cosmological time as in \cite{pinamonti:2011}, we would also have to consider the functional dependence of conformal time on the scale factor.}
  Continuing with the regularized two-point distribution order by order as in proposition~\ref{prop:regularized_state}, we have:

  \begin{proof}[1\st Order]
    Since
    \begin{equation*}
      \gdif \left( \chi_k^0\, \conj\chi{}_k^1 + \chi_k^1\, \conj\chi{}_k^0 + \frac{V}{4 k_0^3} \right)(H, \delta H)(\tau)
      = \frac{1}{4 k_0^3} \int_{\tau_0}^\tau \cos\big( 2 k_0 (\eta - \tau) \big)\, \gdif V'(H, \delta H)(\eta)\, \dif\eta,
    \end{equation*}
    we can proceed with the proof as in proposition~\ref{prop:regularized_state} with $V'$ replaced by $\gdif V'$ and differentiablity follows.
    \noqed
  \end{proof}

  \begin{proof}[2\nd Order]
    As above, this part of the proof can be shown by replacing occurrences of $V$ and $V'$ in \eqref{eq:2nd_order_final} of proposition~\ref{prop:regularized_state} with $\gdif V$ and $\gdif V'$ respectively.
    \noqed
  \end{proof}

  \begin{proof}[Higher Orders]
    For orders $n > 2$ we can again use an estimate similar to \eqref{eq:estimate_chi_n} to obtain a result analogous to~\eqref{eq:higher_order_final}:
    \begin{multline*}
      \abs*{\gdif \left( \sum_{n=3}^\infty \sum_{l=0}^n \chi_k^l\, \conj\chi{}_k^{n-l} \right)(H, \delta H)(\tau)}
      \leq \frac{4}{k_0^4} \left( \int_{\tau_0}^\tau \abs[\big]{\gdif V(H, \delta H)(\eta)}\, \dif\eta \right) \left( \int_{\tau_0}^\tau \abs[\big]{V(\eta)}\, \dif\eta \right)^2
      \\ \times \exp\left( 2 \int_{\tau_0}^\tau (\tau - \eta)\, \abs[\big]{V(\eta)}\, \dif\eta \right).\noqed\qedhere
    \end{multline*}
  \end{proof}
  \noindent In this way we can conclude the proof of the present proposition.
\end{proof}

We can now formulate the main theorem of this paper:
\begin{theorem}\label{thm:main}
  Let $(a_0, H_0), a_0 > 0, \abs{H_0} < H_c$, be some initial conditions fixed at $\tau_0$ for the functional equation \eqref{eq:functional-volterra}.
  There is a non-empty interval $[\tau_0,\tau_1]$ and a closed subset $U \subset C[\tau_0,\tau_1]$ on which a unique solution to \eqref{eq:functional-volterra} exists.
\end{theorem}
\begin{proof}
  In proposition~\ref{prop:gateaux_diffable} we showed that $f$ is continuously Gâteaux differentiable on $\mathcal{U}[\tau_0,\tau_1]$ for any $\tau_1$.
  Using proposition~\ref{prop:closed}, we can thus find a $\tau_1 > \tau_0$ and a closed subset $U \subset \mathcal{U}[\tau_0,\tau_1]$ such that $F(U) \subset U$.
  It then follows from proposition~\ref{prop:contraction} that $F$ has a unique fixed point in $U$.
\end{proof}

Notice that the solution provided by the previous theorem is actually more regular, it is at least differentiable.
Thus the corresponding spacetime is at least $C^2$ and has well-defined curvature tensors.
The extra regularity is provided by the equation \eqref{eq:functional-volterra} and can be easily seen when it is written in differential form.
Unfortunately, it is very difficult to go beyond this regularity, because the employed state is not regular enough to permit the evaluation of higher derivatives of products of fields.

%----------------------------------------------------------------------------%

\section{Global solutions}
\label{sec:global}

In this section we would like to show that it is always possible to extend a `regular' local solution up to the point where either $H^2$ becomes bigger than $H_c^2$ or $a$ diverges.\footnote{$H^2 = H_c^2$ corresponds to a singularity in the derivative of $H$ in the differential form of \eqref{eq:functional-volterra}.}
To this end we start giving a definition we shall use below.

\begin{definition}\label{def:regular}
  A continuous solution $H_*$ of \eqref{eq:functional-volterra} in the interval $[\tau_0,\tau_1]$ with initial conditions
  \begin{equation*}
    a(H_*)(\tau_0) = a_0,
    \qquad
    H_*(\tau_0)^2 = H_0^2 = \frac{1}{3} \big( \rho(\tau_0) + \Lambda \big)
  \end{equation*}
  will be called \emph{regular}, if no singularity for either $a$, $H_*$ or $H_*'$ is encountered in $[\tau_0,\tau_1]$.
  Namely, $H_*$ must satisfy the following conditions:
  \begin{enumerate}
    \item $\norm*{H_*(\tau)}_{C[\tau_0,\tau_1]} < H_c$,
    \item $a_0 \int_{\tau_0}^{\tau} H_*(\eta)\, \dif\eta < 1$ for every $\tau$ in $[\tau_0,\tau_1]$.
  \end{enumerate}
\end{definition}

We remark that a local solution obtained from theorem~\ref{thm:main} is a regular solution.
Henceforth, assume that we have a regular solution $H_*$ as described in the definition.
Notice that condition a) ensures that no singularity in $H_*'$ is met in $[\tau_0,\tau_1]$.
Condition b), on the other hand, ensures that $a$ does not diverge in the interval $[\tau_0,\tau_1]$. Moreover, both a) and b) together imply that $a$ is strictly positive, as can be seen from~\eqref{eq:a_of_H}.

We would like to prove that a regular solution can always be extended in $C[\tau_0,\tau_2]$ for a sufficiently small $\tau_2 - \tau_1 > 0$.
To this end, let us again consider the set
\begin{equation*}
  \mathcal{U}[\tau_1,\tau_2] \defn \Big\{ H \in C[\tau_1, \tau_2] \;\Big|\; \norm{H}_{C[\tau_1,\tau_2]} < \min\big\{ a_1^{-1} (\tau_2 - \tau_1)^{-1}, H_c \big\} \Big\}
\end{equation*}
defined in \eqref{eq:set_U} and where $a_1 \defn a(H_*)(\tau_1)$ is the value assumed by the solution $a(H_*)$ at $\tau_1$.
Now we can give a proposition similar to proposition~\ref{prop:gateaux_diffable}, namely:
\begin{proposition}\label{prop:gateaux_diffable_ext}
  Let $H_*$ be a solution of \eqref{eq:functional-volterra} in $C[\tau_0,\tau_1]$ which is also regular.
  The functional $f(H)$ (see eq.~\eqref{eq:functional_f}), when evaluated on regular extensions of $H_*$ in $\mathcal{U}[\tau_1,\tau_2]$, is continuously Gâteaux differentiable for arbitrary $\tau_2 > \tau_1$.
\end{proposition}
\begin{proof}
  The proof of this proposition can be obtained exactly as the proof of proposition~\ref{prop:gateaux_diffable}.
  However, the estimates we have obtained in lemma~\ref{lem:bounds} and the proof of proposition~\ref{prop:gateaux_diffable} cannot be applied straightforwardly because the state $\omega$ depends on the initial time $\tau_0$ and the initial datum $a_0$ through the construction described in section~\ref{sub:state_construction}.
  Moreover, the estimates of lemma~\ref{lem:bounds} depend on the knowledge of $a$ and $a'$ on the whole interval $[\tau_0,\tau_2]$.
  Luckily enough, we know that the solution $H_*$ is regular in $[\tau_0,\tau_1]$ while we know that the extension restricted to $[\tau_1,\tau_2]$ is in the set $\mathcal{U}[\tau_1,\tau_2]$, thus we just need to use the following estimates
  \begin{equation*}
    \norm{a}_{C[\tau_0,\tau_2]} = \max \big\{ \norm{a}_{C[\tau_1,\tau_2]}, \norm{a}_{C[\tau_0,\tau_1]} \big\},
    \quad
    \norm{a}_{C[\tau_0,\tau_2]}^{-1} = \max \big\{ \norm{a}^{-1}_{[\tau_1,\tau_2]}, \norm{a}^{-1}_{C[\tau_0,\tau_1]} \big\}
  \end{equation*}
  and
  \begin{equation*}
    \norm{a'}_{[\tau_0,\tau_2]} = \max \big\{ \norm{a'}_{[\tau_1,\tau_2]}, \norm{a'}_{[\tau_0,\tau_1]} \big\}.
  \end{equation*}
  With this in mind, we can again use lemma~\ref{lem:bounds} to control the boundedness of $\omega(\norder{\varphi^2})$.
  Then, making the replacements $\tau_0 \to \tau_1, \tau_1 \to \tau_2$ and $a_0 \to a_1$ at the \emph{appropriate} places in proposition~\ref{prop:gateaux_diffable}, one can see that estimates are not substantially influenced and that thesis still holds for $\mathcal{U}[\tau_1,\tau_2]$.
\end{proof}

Notice that it is always possible to fix $\tau_2$ such that $a_1^{-1} (\tau_2 - \tau_1)^{-1} \geq H_c$, whereby $\mathcal{U}[\tau_1,\tau_2]$ becomes the set of \emph{all} possible regular extensions of $H_*$ in $[\tau_1,\tau_2]$.
This guarantees that any extension in $\mathcal{U}[\tau_1,\tau_2]$ is the unique regular extension.

We are now ready to state the main theorem of the present section which can be proven exactly as theorem~\ref{thm:main}.
\begin{theorem}\label{thm:extension}
  Consider a solution $H_*(\tau)$ in $C[\tau_0,\tau_1]$ of the functional equation
  \eqref{eq:functional-volterra}.
  If the solution is regular in $[\tau_0,\tau_1]$, as defined in definition~\ref{def:regular},
  then it is possible to find a $\tau_2 > \tau_1$ such that, the solution $H_*$ can be extended uniquely to $C[\tau_0,\tau_2]$ and the solution is regular therein.
\end{theorem}
\begin{proof}
  Thanks to proposition~\ref{prop:gateaux_diffable_ext}, $f$ is continuously Gâteaux differentiable on all regular extensions of $H_*$ in $\mathcal{U}[\tau_1,\tau_2]$ for any $\tau_2$ such that $a_1^{-1} (\tau_2 - \tau_1)^{-1} \geq H_c$.
  With the remarks of the proof of proposition~\ref{prop:gateaux_diffable_ext} we can use lemma~\ref{lem:bounds} to estimate the boundedness of $\omega(\norder{\varphi^2})$ and apply proposition~\ref{prop:closed} to find a $\tau_2 > \tau_1$ and a closed subset $U \subset \mathcal{U}[\tau_1,\tau_2]$ such that $F(U) \subset U$.
  It then follows from proposition~\ref{prop:contraction} that $F$ has a unique fixed point in $U$.
\end{proof}

We study now the set $\mathcal{S}$ of all possible solutions of \eqref{eq:functional-volterra} which are defined on intervals of the form $[\tau_0,\tau)$, which are regular on any closed interval contained in their domain and which enjoy the same initial values $a_0 = a(\tau_0), H_0 = H(\tau_0)$.
The elements of $\mathcal{S}$ are indicated by $H_I$, where $I$ is the domain of $H_I$ of the form $[\tau_0,\tau_1)$.
We are looking for maximal solutions in this set, where a solution is maximal if it cannot be extended further.
To this end, we notice that it is possible to equip $\mathcal{S}$ with the following partial order relation
\begin{equation*}
  H_I \leq H_J \quad \text{if} \quad I \subset J.
\end{equation*}
Hence, by Zorn's lemma\footnote{It is possible to avoid using Zorn's lemma here and instead prove existence and uniqueness of maximal solution along the lines of the methods discussed in \cite{sbierski:2013}.} applied to the set of all solutions with the given initial values we have the following:
\begin{proposition}
  The maximal regular solution of \eqref{eq:functional-volterra} in $\mathcal{S}$ exists.
\end{proposition}

Since the maximal solution $H_M$ in $\mathcal{S}$ whose domain is $[\tau_0,\tau_M)$ cannot be extended further, this means that its trivial extension on $[\tau_0,\tau_M]$ does not fulfill the hypotheses of theorem~\ref{thm:extension}.
Hence, either its domain is extended until the end of conformal time, \ie, in cosmological time until infinity, or it ceases to be regular on $[\tau_0,\tau_M]$, namely a divergence of $a$ or $H'$ is found in $\tau_M$.

Furthermore, an easy application of the theorem~\ref{thm:extension} permits to obtain the uniqueness of the maximal solution:
\begin{proposition}
  The maximal regular solution of \eqref{eq:functional-volterra} in $\mathcal{S}$ is unique.
\end{proposition}
\begin{proof}
  Suppose to have two maximal solutions, then they must differ in certain closed time-interval $I \subset \RR$.
  Since $I$ is bounded from below by $\tau_0$, the minimum $\tau_{\mathrm{min}}$ of $I$ exists and it is bigger than $\tau_0$.
  But because of theorem~\ref{thm:extension} there is a unique extension after $\tau_{\mathrm{min}}$, hence the two maximal solutions must coincide.
\end{proof}

As for the solution provided by theorem~\ref{thm:main}, also the maximal solution obtained above correspond to a metric with $C^2$ regularity.

%----------------------------------------------------------------------------%

\section{Final comments}
\label{sec:final}

In this paper we have studied the backreaction of a quantum massive scalar field conformally coupled with gravity to cosmological spacetimes.
We have given initial conditions at finite time $\tau=\tau_0$ and we have shown that a unique maximal solution exists. The maximal solution either lasts forever or until a spacetime singularity is reached.

In order to obtain this result, we have used a state which looks as much as possible like the vacuum at the initial time.
Notice that it is possible to choose other classes of states without significantly altering the results obtained in this paper.
In particular, if we restrict ourself to Gaussians pure state which are homogeneous and isotropic, their two-point function takes the form
\begin{equation*}\label{eq:two_point_other}
  \widetilde\omega_2(x, y) = \lim_{\varepsilon \to 0^+} \frac{1}{(2 \uppi)^3} \int_{\RR^3} \frac{\conj\xi_k(\tau_x)}{a(\tau_x)} \frac{\xi{}_k(\tau_y)}{a(\tau_y)}\, \e^{\im \vec{k} \cdot (\vec{x} - \vec{y})} \e^{-\varepsilon k}\, \dif\vec{k},
\end{equation*}
where $\xi_k$ are solutions of \eqref{eq:eom_modes} which enjoy the Wronskian condition $\xi_k' \conj\xi{}_k^{\phantom{'}} - \xi_k^{\phantom{'}} \conj\xi{}_k' = \im$.
These $\chi_k$ can then be written as a Bogoliubov transformation of the modes $\chi_k$ studied earlier in this paper, namely,
\begin{equation*}
  \xi_k = A(k) \chi_k + \conj{B}(k) \conj\chi_k
\end{equation*}
for suitable functions $A$ and $B$.
Then, because of the constraint $\abs{A}^2 - \abs{B}^2 = 1$, the difference
\begin{equation*}
  \widetilde\omega(\norder{\varphi^2}) - \omega(\norder{\varphi^2}) =
  \lim_{\varepsilon \to 0^+} \frac{1}{(2 \uppi)^3} \frac{2}{a^2} \int_{\RR^3} \left(
    \abs{B}^2\chi_k\conj\chi{}_k + \Re{\left(A B\chi_k\chi_k\right)}
  \right) \e^{-\varepsilon k}\, \dif\vec{k}
\end{equation*}
can be easily controlled employing \eqref{eq:estimate_chi_n} if $\abs{B}$ is sufficiently regular (\eg if $B(k)$ is in $L^2 \cap L^1$).\footnote{A detailed analysis of this problem is present in \cite{zschoche:2013}.}
With this observation it is possible to obtain again all the estimates used in the proofs of the theorems \ref{thm:main} and~\ref{thm:extension}.

In the future, it would be desirable to study the semiclassical equations in more general cases, namely for more general fields, abandoning for example the conformal coupling, and for more general background geometries.
The results presented here cannot straightforwardly be extended to fields which are not conformally coupled to curvature or to spacetimes that are not conformally flat because in that case fourth order derivatives of the metric originating in the conformal anomaly cannot be cancelled by a judicious choice of renormalization parameters, \ie, Wald's fifth axiom \cite{wald:1977} cannot be satisfied.
To still solve the semiclassical Einstein equation with methods similar to those presented here,
a deeper analysis of the states is required, in particular, one needs states of higher regularity.
A preliminary study in this direction can be found in a paper of Eltzner and Gottschalk \cite{eltzner:2011}, where the semiclassical Einstein equation on a FLRW background with non-conformally coupled scalar field is discussed.
The case of backgrounds which are only spherically symmetric is interesting from many perspectives.
Its analysis could give new hints on the problem of semiclassical black hole evaporation and confirm the nice two-dimensional results obtained in \cite{ashtekar:2010}.
Finally, the limit of validity of the employed equation needs to be carefully addressed in the future.

%----------------------------------------------------------------------------%

\begin{acknowledgments}
  We would like to thank T.-P. Hack and V. Moretti for helpful discussions.
  The work of N.P. has been supported partly by the Indam-GNFM project ``Influenza della materia quantistica sulle fluttuazioni gravitazionali'' (2013).
\end{acknowledgments}

%----------------------------------------------------------------------------%

\appendix

\section{Banach fixed-point theorem}
\label{sec:appendix}

In this appendix we will review the Banach fixed-point theorem and a few related results on contraction maps and functional derivatives.
Let us begin by stating (without proof) the famous Banach fixed-point theorem:

\begin{theorem}[Banach fixed-point theorem]
  Let $f: U \to U$ be a contraction on a (non-empty) complete metric space $U$.
  Then $f$ has a unique fixed-point $x = f(x)$.
  Furthermore, taking an arbitrary initial value $x_0 \in U$, $x$ is the limit of the sequence $\{x_n\}$ defined by the iterative procedure $x_{n+1} = f(x_n)$.
\end{theorem}

In practice it is often difficult to prove directly that a given map is a contraction.
Therefore we will enlighten in the following some sufficient conditions which guarantee that a given map of the type studied in this paper is a contraction map.
First, let us define a directional derivative, the Gâteaux derivative, according to~\cite{hamilton:1982}:
\begin{definition}\label{def:gateaux}
  Let $V, W$ be Fréchet spaces\footnote{The Gâteaux derivative can also be generalized to locally convex topological vector spaces.} and $U \subset V$ an open subset.
  The \emph{Gâteaux derivative} of a map $f: U \to W$ at the point $x \in U$ in the direction $h \in V$ is defined as the map $\gdif f: U \times V \to W$,
  \begin{equation*}
    \gdif f(x, h) \defn \lim_{\varepsilon \to 0} \frac{f(x + \varepsilon\, h) - f(x)}{\varepsilon} = \left. \od{}{\varepsilon}\, f(x + \varepsilon\, h) \right|_{\varepsilon=0}.
  \end{equation*}
  The map $f$ is called \emph{Gâteaux differentiable at $x$} if the limit exists for all $h \in V$.
  It is called \emph{continuously Gâteaux differentiable} if $\gdif f$ is continuous on the product space $U \times V$ for all $x$ and $h$.
\end{definition}

Among the properties of this derivative discussed in \cite{hamilton:1982}, we find:
\begin{proposition}\label{prop:gateaux_properties}
  Let $f: U \subset V \to W$ be a continuously Gâteaux differentiable map between Fréchet spaces $V, W$. Then:
  \begin{enumerate}
    \item $\gdif f(x, h)$ is linear in $h$,
    \item
      the fundamental theorem of calculus holds if $U$ is convex, \ie,
      \begin{equation*}
        f(x + h) - f(x) = \int_0^1 \gdif f(x + t\, h, h)\, \dif t.
      \end{equation*}
  \end{enumerate}
\end{proposition}

Therefore, specializing to Banach spaces, one can show (\cf the more general result in~\cite[Theorem 5.1.3]{hamilton:1982}):
\begin{proposition}\label{prop:lipschitz}
  Let $f: U \subset V \to W$ be a continuously Gâteaux differentiable map between Banach spaces $V, W$ with the norms $\norm{\,\cdot\,}_V$ and $\norm{\,\cdot\,}_W$ respectively.
  Then $f$ is \emph{locally Lipschitz}, that is, for every convex neighborhood $\mathcal{N}$ of $x_0 \in U$ there exists a $K \geq 0$ such that for all $x_1, x_2 \in \mathcal{N}$
  \begin{equation*}
    \norm{f(x_1) - f(x_2)}_W \leq K\, \norm{x_1 - x_2}_V.
  \end{equation*}
\end{proposition}
\begin{proof}
  Since the derivative $\gdif f(x, h)$ is continuous and linear in $h$, there exists a (convex) neighborhood $\mathcal{N}$ of $x_0$ such that\footnote{The operator norm $\opnorm{A}$ of a linear operator $A: V \to W$ between two normed vector spaces $V, W$ is defined as $\opnorm{A} \defn \sup \{ \norm{A(x)}_W \mid x \in V \text{ with } \norm{x}_V \leq 1 \}$}
  \begin{equation*}
    \norm{\gdif f(x, h)}_W \leq \opnorm{\gdif f(x)} \norm{h}_V \leq K\, \norm{h}_V
  \end{equation*}
  for all $x \in \mathcal{N}$.
  As \emph{Lipschitz constant} $K$ we can choose the supremum of $x \mapsto \opnorm{\gdif f(x)}$ in $\mathcal{N}$.
  By the fundamental theorem of calculus we have for $x_1, x_2 \in \mathcal{N}$
  \begin{equation*}
    f(x_1) - f(x_2) = \int_0^1 \gdif f\big( x_2 + t\, (x_1 - x_2), x_1 - x_2 \big)\, \dif t.
  \end{equation*}
  Hence, taking the norm on both sides, the previous equation yields
  \begin{equation*}
    \norm{f(x_1) - f(x_2)}_W
    \leq \int_0^1 \norm*{\gdif f\big( x_2 + t\, (x_1 - x_2), x_1 - x_2 \big)}_W\, \dif t
    \leq K\, \norm{x_1 - x_2}_V.
    \qedhere
  \end{equation*}
\end{proof}

The Gâteaux derivative is closely related to the Fréchet derivative:
\begin{definition}\label{def:frechet}
  Again, let $V, W$ be Banach spaces and $U \subset V$ an open subset.
  A map $f: U \to W$ is called \emph{Fréchet differentiable at $x \in U$} if there exists a bounded linear operator $\fdif f(x, h)$, the \emph{Fréchet derivative} of $f$ at $x$, such that
  \begin{equation*}
    \lim_{\norm{h} \to 0} \frac{\norm{f(x + h) - f(x) - \fdif f(x, h)}_W}{\norm{h}_V} = 0.
  \end{equation*}
\end{definition}
\begin{proposition}\label{prop:gateaux_frechet}
  Given the definitions above, if $f$ is continuously Gâteaux differentiable in a neighborhood of $x \in U$, then $f$ is also Fréchet differentiable at $x$ and the two derivatives agree.
\end{proposition}
\begin{proof}
  As in proposition~\ref{prop:lipschitz}, since the derivative $\gdif f(z)$ is a continuous and linear operator, there exists a (convex) neighborhood $\mathcal{N}$ of $x$ where it is bounded.
  Using the fundamental theorem of calculus again, we obtain for $h, z \in \mathcal{N}$
  \begin{equation*}
    \norm{f(x + h) - f(x) - \gdif f(z, h)}_W \leq \sup_{t \in [0,1]} \opnorm{\gdif f(x + t h) - \gdif f(z)} \norm{h}_V.
  \end{equation*}
  In particular this holds for $z = x$ and thus $f$ is Fréchet differentiable at $x$ with $\fdif f(x) = \gdif f(x)$.
\end{proof}

Henceforth we will restrict our attention to Banach spaces $C[a, b]$ of continuous functions on an interval $[a, b]$ equipped with the \emph{uniform norm}
\begin{equation*}
  \norm{X}_{C[a,b]} \defn \norm{X}_\infty \defn \sup_{t \in [a,b]} \abs{X(t)},
\end{equation*}
where we will use $\norm{X}_{C[a,b]}$ instead of the more common $\norm{X}_\infty$ to emphasize the interval over which the supremum is taken.
Moreover, we will study maps $F: C[a, b] \to C[a, b]$ of the form (compare this to~\eqref{eq:functional-volterra-short})
\begin{equation}\label{eq:F-def}
  F(X)(t) \defn F_0(t) + \int_{a}^t f(X)(s)\, \dif s,
\end{equation}
where $F_0 \in C[a, b]$ and $f: C[a, b] \to C[a, b]$ is a `retarded' functional which satisfies $f(X + Y)(s) = f(X)(s), s \leq t$, for all $X, Y \in C[a, b]$ with $Y$ compactly supported in $(t,b]$.
It is obvious from \eqref{eq:F-def} that also $F$ is a retarded functional.

\begin{lemma}\label{lem:retarded_functional}
  The retarded functional $f$ can be restricted to a map $f_t$ on $C[a, t], t \in [a,b]$, \ie, $f(X)(s) = f_t(X \restriction_{[a,t]})(s)$ for every $X \in C[a,b]$ and $s \in [a,t]$.
  If $f$ is (continuously) Gâteaux differentiable, then so is $f_t$.
\end{lemma}
\begin{proof}
  Any function $X_t \in C[a,t]$ can be continuously extended to a function $X \in C[a,b]$.
  We can then define $f_t$ as $f_t(X_t) = f(X) \restriction_{[a,t]}$ independently of the chosen extension.
  From the definition~\ref{def:gateaux} it is then obvious that $f_t$ will be (continuously) Gâteaux differentiable if this is true for $f$.
\end{proof}

Of course the same holds true also for the functional $F$.
From now on we shall denote the restriction of such functionals to smaller intervals by the same symbol as the original functional.

We are now ready to present two propositions which will be used in the proof of the main theorems:

\begin{proposition}\label{prop:closed}
  Suppose that $f$ is bounded on a set $\mathcal{U} \subset C[a,b]$ which also includes a neighborhood $\mathcal{N}$ of $F_0$ defined as $\mathcal{N}=\{ X \mid \norm{X-F_0}_{C[a,b]} <\delta \}$ for some $\delta$.
  Let $\mathcal{U}_t = \{ X \restriction_{[a,t]} \mid X \in \mathcal{U} \}$ be the restriction of $\mathcal{U}$ to $[a,t]$ and $\mathcal{N}_t$ that of $\mathcal{N}$ to $[a,t]$.
  Then there exists $t \in (a,b]$ such that the restriction of $F$ to $C[a,t]$ satisfies $F(U) \subset U$ for all $U \subset \mathcal{U}_{t}$ that contain the neighborhood $\mathcal{N}_t$ of $F_0$.
\end{proposition}
\begin{proof}
  If $f$ is bounded on $\mathcal{U}$, then its restriction to $C[a,t]$ is bounded on $\mathcal{U}_t$.
  That is, for all $X \in \mathcal{U}_t$
  \begin{equation*}
    \norm{f(X)}_{C[a,t]} \leq K_t
    = \sup_{Y \in \mathcal{U}_t} \, \norm{f(Y)}_{C[a,t]}
    \leq \sup_{Y \in \mathcal{U}_t} \, \norm{f(Y)}_{C[a,b]}.
  \end{equation*}
  Then, taking the norm of \eqref{eq:F-def} after subtracting $F_0$, one obtains
  \begin{equation*}
    \norm{F(X)-F_0}_{C[a,t]}
    \leq  (t - a) \norm{f(X)}_{C[a,t]}
    \leq  (t - a) K_t.
  \end{equation*}
  Since $K_t \leq K_{t'}$ for $t < t'$ we can always find a $t$ such that $F(X)$ stays close enough to $F_0$ to be included in $\mathcal{N}_t$ and hence in $\mathcal{U}_t$.
  Assume now that $U \subset \mathcal{U}_t$ contains the neighborhood $\mathcal{N}_t$.
  Then the thesis follows because $\sup_{Y \in U} \, \norm{f(Y)}_{C[a,t]} \leq K_t$.
\end{proof}

\begin{proposition}\label{prop:contraction}
  Let $U \subset \mathcal{U} \subset C[a,b]$ be convex sets with $\mathcal{U}$ open and $U$ closed such that $F(U) \subset U$.
  If $f$ is continuously Gâteaux differentiable on $\mathcal{U}$, then there exists $N \in \NN$ such that $F^n$ is a contraction map on $U$ for all $n \geq N$, that is, there exists $C \in (0,1)$ for all $X, Y \in U$ such that
  \begin{equation*}\label{eq:contraction}
    \norm{F^n(X) - F^n(Y)}_{C[a,b]} \leq C\, \norm{X - Y}_{C[a,b]}.
  \end{equation*}
  In consequence $F$ will have a unique fixed-point in $U$.
\end{proposition}
\begin{proof}
  Suppose that, for all $n$ and arbitrary $t \in [a,b]$,
  \begin{equation}\label{eq:Fn-step}
    \norm{F^n(X) - F^n(Y)}_{C[a,t]}
    \leq \frac{K^n (t - a)^n}{n!} \norm{X - Y}_{C[a,t]}.
  \end{equation}
  Setting $t = b$ and choosing a sufficiently large $N$, this implies that the map $F^n$ is contractive for every $n\geq N$.
  Since $U$ is a complete metric space (it is a closed subset of a Banach space) and $F(U) \subset U$, we can apply Banach's fixed point theorem for $F^{n}$ and $F^{n+1}$ to find that also $F$ has a unique fixed point in $U$.

  We can show the statement~\eqref{eq:Fn-step} using an inductive procedure.
  Applying proposition~\ref{prop:lipschitz} and lemma~\ref{lem:retarded_functional}, we find that $f$ is locally Lipschitz as a functional on $\mathcal{U}$ and its restrictions $\mathcal{U}_t = \{ Z \restriction_{[a,t]} \mid Z \in \mathcal{U} \}$ with common Lipschitz constant $K = \sup_{X \in \mathcal{U}} \opnorm{\gdif f(X)}$.
  Using the uniform norm on $C[a,t]$, we thus obtain
  \begin{equation*}
    \norm{F(X) - F(Y)}_{C[a,t]}
    \leq \int_a^t \norm{f(X) - f(Y)}_{C[a,t]}\, \dif s
    \leq K (t - a) \norm{X - Y}_{C[a,t]},
  \end{equation*}
  which proves \eqref{eq:Fn-step} for $n = 1$.
  Suppose now that \eqref{eq:Fn-step} holds for $n$.
  Then,
  \begin{align*}
    \abs{F^{n+1}(X)(t) - F^{n+1}(Y)(t)}
    & \leq \int_a^t \norm*{f\big(F^n(X)\big) - f\big(F^n(Y)\big)}_{C[a,s]}\, \dif s \\
    & \leq K \int_a^t \norm*{F^n(X) - F^n(Y)}_{C[a,s]}\, \dif s \\
    & \leq \frac{K^{n+1}}{n!} \int_a^t (s - a)^n \norm{X - Y}_{C[a,s]}\, \dif s \\
    & \leq \frac{K^{n+1} (t - a)^{n+1}}{(n+1)!} \norm{X - Y}_{C[a,t]}
  \end{align*}
  and, since the estimate on the right-hand side is monotonically increasing in $t$,
  \begin{equation*}
    \norm{F^{n+1}(X) - F^{n+1}(Y)}_{C[a,t]} \leq \frac{K^{n+1} (t - a)^{n+1}}{(n+1)!} \norm{X - Y}_{C[a,t]},
  \end{equation*}
  which implies that \eqref{eq:Fn-step} holds also for $n+1$, thus concluding the proof.
\end{proof}

%----------------------------------------------------------------------------%


\small
\begin{thebibliography}{35}
\providecommand{\natexlab}[1]{#1}
\providecommand{\url}[1]{\texttt{#1}}
\providecommand{\urlprefix}{URL }
\providecommand{\eprint}[2][]{\url{#2}}

\bibitem[Anderson(1983)]{anderson:1983}
Anderson, P.R.: Effects of quantum fields on singularities and particle
  horizons in the early universe.
\newblock \href{http://dx.doi.org/10.1103/PhysRevD.28.271}{Physical Review D
  \textbf{28}, 271--285 (1983)}

\bibitem[Anderson(1984)]{anderson:1984}
Anderson, P.R.: Effects of quantum fields on singularities and particle
  horizons in the early universe. {II}.
\newblock \href{http://dx.doi.org/10.1103/PhysRevD.29.615}{Physical Review D
  \textbf{29}, 615--627 (1984)}

\bibitem[Anderson(1985)]{anderson:1985}
Anderson, P.R.: Effects of quantum fields on singularities and particle
  horizons in the early universe. {III.} {The} conformally coupled massive
  scalar field.
\newblock \href{http://dx.doi.org/10.1103/PhysRevD.32.1302}{Physical Review D
  \textbf{32}, 1302--1315 (1985)}

\bibitem[Anderson(1986)]{anderson:1986}
Anderson, P.R.: Effects of quantum fields on singularities and particle
  horizons in the early universe. {IV}. {Initially} empty universes.
\newblock \href{http://dx.doi.org/10.1103/PhysRevD.33.1567}{Physical Review D
  \textbf{33}, 1567--1575 (1986)}

\bibitem[Ashtekar, Pretorius, and Ramazanoglu(2011)]{ashtekar:2010}
Ashtekar, A., Pretorius, F., Ramazanoglu, F.M.: Surprises in the Evaporation of
  2-Dimensional Black Holes.
\newblock \href{http://dx.doi.org/10.1103/PhysRevLett.106.161303}{Physical
  Review Letters \textbf{106}, 161303 (2011)}.
\newblock \href{http://arxiv.org/abs/1011.6442}{{\ttfamily arXiv:1011.6442
  [gr-qc]}}

\bibitem[B{\"a}r, Ginoux, and Pf{\"a}ffle(2007)]{bar:2007}
B{\"a}r, C., Ginoux, N., Pf{\"a}ffle, F.:
  \href{http://dx.doi.org/10.4171/037}{\textit{Wave equations on {Lorentzian}
  manifolds and quantization}}.
\newblock ESI Lectures in Mathematical Physics. European Mathematical Society
  (2007).
\newblock \href{http://arxiv.org/abs/0806.1036}{{\ttfamily arXiv:0806.1036
  [math.DG]}}

\bibitem[Brunetti and Fredenhagen(2000)]{brunetti:2000}
Brunetti, R., Fredenhagen, K.: Microlocal analysis and interacting quantum
  field theories: {Renormalization} on physical backgrounds.
\newblock \href{http://dx.doi.org/10.1007/s002200050004}{Communications in
  Mathematical Physics \textbf{208}, 623--661 (2000)}.
\newblock \href{http://arxiv.org/abs/math-ph/9903028}{{\ttfamily
  arXiv:math-ph/9903028}}

\bibitem[Brunetti, Fredenhagen, and K{\"o}hler(1996)]{brunetti:1996}
Brunetti, R., Fredenhagen, K., K{\"o}hler, M.: The microlocal spectrum
  condition and {Wick} polynomials of free fields on curved spacetimes.
\newblock \href{http://dx.doi.org/10.1007/BF02099626}{Communications in
  Mathematical Physics \textbf{180}, 633--652 (1996)}

\bibitem[Brunetti, Fredenhagen, and Verch(2003)]{brunetti:2003}
Brunetti, R., Fredenhagen, K., Verch, R.: The generally covariant locality
  principle -- a new paradigm for local quantum field theory.
\newblock \href{http://dx.doi.org/10.1007/s00220-003-0815-7}{Communications in
  Mathematical Physics \textbf{237}, 31--68 (2003)}.
\newblock \href{http://arxiv.org/abs/math-ph/0112041}{{\ttfamily
  arXiv:math-ph/0112041}}

\bibitem[Bunch and Davies(1978)]{bunch:1978}
Bunch, T.S., Davies, P.C.W.: Quantum field theory in de {Sitter} space:
  Renormalization by point-splitting.
\newblock \href{http://dx.doi.org/10.1098/rspa.1978.0060}{Proceedings of the
  Royal Society of London. A. Mathematical and Physical Sciences \textbf{360},
  117--134 (1978)}

\bibitem[Dappiaggi, Fredenhagen, and Pinamonti(2008)]{dappiaggi:2008}
Dappiaggi, C., Fredenhagen, K., Pinamonti, N.: Stable cosmological models
  driven by a free quantum scalar field.
\newblock \href{http://dx.doi.org/10.1103/PhysRevD.77.104015}{Physical Review D
  \textbf{77}, 104015 (2008)}.
\newblock \href{http://arxiv.org/abs/0801.2850}{{\ttfamily arXiv:0801.2850
  [gr-qc]}}

\bibitem[Dappiaggi, Moretti, and
  Pinamonti(2009{\natexlab{a}})]{dappiaggi:2009b}
Dappiaggi, C., Moretti, V., Pinamonti, N.: Cosmological horizons and
  reconstruction of quantum field theories.
\newblock \href{http://dx.doi.org/10.1007/s00220-008-0653-8}{Communications in
  Mathematical Physics \textbf{285}, 1129--1163 (2009{\natexlab{a}})}.
\newblock \href{http://arxiv.org/abs/0712.1770}{{\ttfamily arXiv:0712.1770
  [gr-qc]}}

\bibitem[Dappiaggi, Moretti, and
  Pinamonti(2009{\natexlab{b}})]{dappiaggi:2009a}
Dappiaggi, C., Moretti, V., Pinamonti, N.: Distinguished quantum states in a
  class of cosmological spacetimes and their {Hadamard} property.
\newblock \href{http://dx.doi.org/10.1063/1.3122770}{Journal of Mathematical
  Physics \textbf{50}, 062304 (2009{\natexlab{b}})}.
\newblock \href{http://arxiv.org/abs/0812.4033}{{\ttfamily arXiv:0812.4033
  [gr-qc]}}

\bibitem[Eltzner and Gottschalk(2011)]{eltzner:2011}
Eltzner, B., Gottschalk, H.: Dynamical backreaction in {Robertson}--{Walker}
  spacetime.
\newblock \href{http://dx.doi.org/10.1142/S0129055X11004357}{Reviews in
  Mathematical Physics \textbf{23}, 531--551 (2011)}.
\newblock \href{http://arxiv.org/abs/1003.3630}{{\ttfamily arXiv:1003.3630
  [math-ph]}}

\bibitem[Fewster and Verch(2012)]{fewster:2012}
Fewster, C.J., Verch, R.: Dynamical locality and covariance: What makes a
  physical theory the same in all spacetimes?
\newblock \href{http://dx.doi.org/10.1007/s00023-012-0165-0}{Annales Henri
  Poincar{\'e} \textbf{13}, 1613--1674 (2012)}.
\newblock \href{http://arxiv.org/abs/1106.4785}{{\ttfamily arXiv:1106.4785
  [math-ph]}}

\bibitem[Haag(1996)]{haag:1996}
Haag, R.: \href{http://www.worldcat.org/isbn/3540614516}{\textit{Local quantum
  physics}}.
\newblock Springer, 2nd edn. (1996)

\bibitem[Hack(2010)]{hack:2010}
Hack, T.P.: On the backreaction of scalar and spinor quantum fields in curved
  spacetimes.
\newblock Ph.D. thesis, Universit{\"a}t Hamburg (2010).
\newblock \href{http://arxiv.org/abs/1008.1776}{{\ttfamily arXiv:1008.1776
  [gr-qc]}}

\bibitem[Hack(2013)]{hack:2013}
Hack, T.P.: The {$\upLambda$CDM}-model in quantum field theory on curved
  spacetime and dark radiation (2013).
\newblock \href{http://arxiv.org/abs/1306.3074}{{\ttfamily arXiv:1306.3074
  [gr-qc]}}

\bibitem[Hamilton(1982)]{hamilton:1982}
Hamilton, R.S.: The inverse function theorem of {Nash} and {Moser}.
\newblock Bulletin of the American Mathematical Society \textbf{7}, 65--222
  (1982)

\bibitem[Hollands and Wald(2001)]{hollands:2001}
Hollands, S., Wald, R.M.: Local {Wick} polynomials and time ordered products of
  quantum fields in curved spacetime.
\newblock \href{http://dx.doi.org/10.1007/s002200100540}{Communications in
  Mathematical Physics \textbf{223}, 289--326 (2001)}.
\newblock \href{http://arxiv.org/abs/gr-qc/0103074}{{\ttfamily
  arXiv:gr-qc/0103074}}

\bibitem[Junker and Schrohe(2002)]{junker:2002}
Junker, W., Schrohe, E.: Adiabatic vacuum states on general spacetime
  manifolds: Definition, construction, and physical properties, construction,
  and physical properties.
\newblock \href{http://dx.doi.org/10.1007/s000230200001}{Annales Henri
  Poincar{\'e} \textbf{3}, 1113--1181 (2002)}.
\newblock \href{http://arxiv.org/abs/math-ph/0109010}{{\ttfamily
  arXiv:math-ph/0109010}}

\bibitem[Kofman, Linde, and Starobinsky(1985)]{kofman:1985}
Kofman, L.A., Linde, A.D., Starobinsky, A.A.: Inflationary Universe Generated
  by the Combined Action of a Scalar Field and Gravitational Vacuum
  Polarization.
\newblock \href{http://dx.doi.org/10.1016/0370-2693(85)90381-8}{Physics Letters
  \textbf{B157}, 361--367 (1985)}

\bibitem[Kusku(2008)]{kusku:2008}
Kusku, M.: \href{http://dx.doi.org/10.3204/DESY-THESIS-2008-020}{A Class of
  Almost Equilibrium States in {Robertson}--{Walker} Spacetimes}.
\newblock Ph.D. thesis, Universit{\"a}t Hamburg (2008)

\bibitem[L{\"u}ders and Roberts(1990)]{luders:1990}
L{\"u}ders, C., Roberts, J.E.: Local quasiequivalence and adiabatic vacuum
  states.
\newblock \href{http://dx.doi.org/10.1007/BF02102088}{Communications in
  Mathematical Physics \textbf{134}, 29--63 (1990)}

\bibitem[Moretti(2003)]{moretti:2003}
Moretti, V.: Comments on the stress-energy tensor operator in curved spacetime.
\newblock \href{http://dx.doi.org/10.1007/s00220-002-0702-7}{Communications in
  Mathematical Physics \textbf{232}, 189--221 (2003)}.
\newblock \href{http://arxiv.org/abs/gr-qc/0109048}{{\ttfamily
  arXiv:gr-qc/0109048}}

\bibitem[Olbermann(2007)]{olbermann:2007}
Olbermann, H.: States of low energy on {Robertson}--{Walker} spacetimes.
\newblock \href{http://dx.doi.org/10.1088/0264-9381/24/20/007}{Classical and
  Quantum Gravity \textbf{24}, 5011--5030 (2007)}.
\newblock \href{http://arxiv.org/abs/0704.2986}{{\ttfamily arXiv:0704.2986
  [gr-qc]}}

\bibitem[Parker(1969)]{parker:1969}
Parker, L.E.: Quantized fields and particle creation in expanding universes.
  {I}.
\newblock \href{http://dx.doi.org/10.1103/PhysRev.183.1057}{Physical Review
  \textbf{183}, 1057--1068 (1969)}

\bibitem[Pinamonti(2011)]{pinamonti:2011}
Pinamonti, N.: On the initial conditions and solutions of the semiclassical
  {Einstein} equations in a cosmological scenario.
\newblock \href{http://dx.doi.org/10.1007/s00220-011-1268-z}{Communications in
  Mathematical Physics \textbf{305}, 563--604 (2011)}.
\newblock \href{http://arxiv.org/abs/1001.0864}{{\ttfamily arXiv:1001.0864
  [gr-qc]}}

\bibitem[Radzikowski(1996)]{radzikowski:1996}
Radzikowski, M.J.: Micro-local approach to the {Hadamard} condition in quantum
  field theory on curved space-time.
\newblock \href{http://dx.doi.org/10.1007/BF02100096}{Communications in
  Mathematical Physics \textbf{179}, 529--553 (1996)}

\bibitem[Sahlmann and Verch(2001)]{sahlmann:2001}
Sahlmann, H., Verch, R.: Microlocal spectrum condition and {Hadamard} form for
  vector-valued quantum fields in curved spacetime.
\newblock \href{http://dx.doi.org/10.1142/S0129055X01001010}{Reviews in
  Mathematical Physics \textbf{13}, 1203--1246 (2001)}.
\newblock \href{http://arxiv.org/abs/math-ph/0008029}{{\ttfamily
  arXiv:math-ph/0008029}}

\bibitem[Sbierski(2013)]{sbierski:2013}
Sbierski, J.: On the existence of a maximal {Cauchy} development for the
  {Einstein} equations -- a dezornification  (2013).
\newblock \href{http://arxiv.org/abs/1309.7591}{{\ttfamily arXiv:1309.7591
  [gr-qc]}}

\bibitem[Starobinsky(1980)]{starobinsky:1980}
Starobinsky, A.A.: A new type of isotropic cosmological models without
  singularity.
\newblock \href{http://dx.doi.org/10.1016/0370-2693(80)90670-X}{Physics Letters
  B \textbf{91}, 99--102 (1980)}

\bibitem[Wald(1977)]{wald:1977}
Wald, R.M.: The back reaction effect in particle creation in curved spacetime.
\newblock \href{http://dx.doi.org/10.1007/BF01609833}{Communications in
  Mathematical Physics \textbf{54}, 1--19 (1977)}

\bibitem[Wald(1978)]{wald:1978}
Wald, R.M.: Trace anomaly of a conformally invariant quantum field in curved
  spacetime.
\newblock \href{http://dx.doi.org/10.1103/PhysRevD.17.1477}{Physical Review D
  \textbf{17}, 1477--1484 (1978)}

\bibitem[Zschoche(2013)]{zschoche:2013}
Zschoche, J.: The {Chaplygin} Gas Equation of State for the Quantized Free
  Scalar Field on Cosmological Spacetimes.
\newblock \href{http://dx.doi.org/10.1007/s00023-013-0281-5}{Annales Henri
  Poincar{\'e} \textbf{Online} (2013)}.
\newblock \href{http://arxiv.org/abs/1303.4992}{{\ttfamily arXiv:1303.4992
  [gr-qc]}}

\end{thebibliography}
\end{document}